\documentclass[reprint,superscriptaddress,amsmath,amssymb,aps,pre]{revtex4-1}

\usepackage{amsmath}
\usepackage{amsthm}
\usepackage{graphicx} 
\usepackage{dcolumn}  
\usepackage{bm}       
\usepackage{hyperref} 
\usepackage{xspace}

\usepackage{color}
\usepackage[normalem]{ulem}

\newcommand{\VEC}[1]{\mbox{\boldmath${#1}$}}
\newcommand{\schr}{Schr\"odinger\xspace}
\newcommand{\ps}{phase space\xspace}
\newcommand{\ns}{non\-separable\xspace}

\newcommand{\ep}{\varepsilon}

\newcommand{\be}{\begin{equation}}
\newcommand{\ee}{\end{equation}}

\newcommand{\nn}{\nonumber}
\newcommand{\hT}{\hat{\VEC T}}
\newcommand{\hV}{\hat{\VEC V}}
\newcommand{\hH}{\hat{\VEC H}}

\newcommand{\hx}{\hat{\VEC x}}
\newcommand{\hp}{\hat{\VEC p}}
\newcommand{\hU}{\hat{\VEC U}}

\newcommand{\htH}{\hat{\VEC H}_A}

\newcommand{\new}[1]{{#1}}
\def\t1{e_{_T}}
\def\v1{e_{_V}}

\def\ct{e_{_{TTV}}}
\def\cv{e_{_{VTV}}}

\newtheorem{theorem}{Theorem}

\begin{document}

\title{Exponential unitary integrators for nonseparable quantum Hamiltonians}

\author{Maximilian \'Ciri\'c}
\email{max.ciric@outlook.com}
\affiliation{Department of Physics,~Astronomy~and~Mathematics,
    University~of Hertfordshire, Hatfield, AL10 9AB, UK}

\author{Denys I. Bondar}
\email{dbondar@tulane.edu}
\affiliation{Tulane University, New Orleans, LA 70118, USA}

\author{Ole Steuernagel}
\email{Ole.Steuernagel@gmail.com}
\affiliation{Department of Physics,~Astronomy~and~Mathematics,
    University~of Hertfordshire, Hatfield, AL10 9AB, UK}

\date{\today}

\begin{abstract}
  Quantum Hamiltonians containing nonsepa\-ra\-ble products of non-commuting operators, such as
  $\hx^m \hp^n$, are problematic for numerical studies using split-operator techniques since such
  products cannot be represented as a sum of separable terms, such as $T(\hp) + V(\hx)$. In the case
  of classical physics, Chin {[Phys. Rev. E {\bf 80}, 037701 (2009)]} developed a procedure to
  approximately represent nonsepa\-ra\-ble terms in terms of separable ones. We extend Chin's idea
  to quantum systems. We demonstrate our findings by numerically evolving the Wigner distribution of
  a Kerr-type oscillator whose Hamiltonian contains the nonsepa\-ra\-ble term
  $\hx^2 \hp^2+\hp^2 \hx^2$.  The general applicability of Chin's approach to any Hamiltonian of
  polynomial form is proven.
\end{abstract}

\maketitle

\section{\label{sec:introduction}Introduction}

Split-operator methods~\cite{Feit_JCP82} are popular across many domains of physics because they
combine the best of two worlds -- simplicity of implementation and preservation of physical
properties such as norm and energy. For the time evolution of Hamiltonian quantum systems (including
those with nonlinearities~\cite{Javanainen_JPA06}), unitary split-operator
integrators~\cite{Yoshida_PLA90} have emerged as reliable workhorses. However, such split-operator
me\-thods are currently restricted to Hamiltonians which are separable~\cite{Cabrera_PRA15}, i.e.
those that are a sum of two terms $T(\hp)$ and $V(\hx)$, each depending only on~$\hp$
and~$\hx$ respectively (throughout, we denote all quantum operators in bold face).

Other recent approaches treating nonsepa\-ra\-ble classical Hamiltonians exist, including one that
yields good long time behaviour for classical dynamics, but which does this at the resource-intensive
price of doubling up the \ps, see~\cite{Tao_PRE16} and references therein. In the case of quantum
systems this is too high a price to pay.

Classical evolution is governed by Poisson-brackets, whose commutators Chin's
method~\cite{Chin_PRE07} combines algebraically such that upon application of exponential
split-operator integrators it extends to the treatment of \ns Hamiltonians. Whilst such classical
evolution is \ps volume-conserving~\cite{Chin_AJP20} quantum evolution is not~\cite{Oliva_PhysA17},
and its evolution is governed by Moyal brackets. Therefore it remained unclear whether Chin's
approach can be adapted to quantum systems.

Here we show that Chin's approach can be extended to quantum systems, see
Sect.~\ref{sec:QuantumChin}, such as Kerr-oscillators, see Sect.~\ref{sec:Kerr}. In
Sect.~\ref{sec:ChinGeneral}, we generalize its application to (\ns) Hamiltonians composed of
polyno\-mials. We use Wigner's quantum \ps representation~\cite{Curtright_Treatise2014} and
investigate its numerical performance in Secs.~\ref{sec:WignerKerr} and~\ref{sec:Numerics} and then
conclude in Sect.~\ref{sec:conclusion}.

\section{Extending Chin's approach to quantum evolution}\label{sec:QuantumChin}

\subsection{\label{subsec:separable_to_nonseparable}From separable to \ns propa\-ga\-tors}

Separable Hamiltonians of the form $\hH = T(\hp) + V(\hx)$ allow for operator splitting.
\new{ Here, we deal with classical and quantum operators as well as their eigenvalues.
  Following~\cite{Cabrera_PRA15}, we therefore adopt the following notation. Throughout, bold
  lettering ($\bm x, \bm p$) refers to quantum scalars whereas regular lettering ($ x, p$) refers to
  their classical counterparts or functions, such as~$T$ or $V.$ Hats [($\hp$, $\hx$) vs. ($\hat p$,
  $\hat x$)] indicate their respective operators.}

A state is propagated through a small time-step, $\ep = -i \Delta t/\hbar$, by the unitary
propaga\-tor~$\hU(\ep)$. $\hU$ can be split into the approximate
form~\cite{Chin_PRE07,Yoshida_PLA90}
\begin{equation}\label{eq:U_time_ev}
\hU(\ep) = {\rm e}^{\ep \hH} = {\rm e}^{\ep(\hT+\hV)} \approx \prod_{i=1}^N
  {\rm e}^{t_i\ep \hT}{\rm e}^{v_i\ep \hV} ={\rm e}^{\ep\htH},
\end{equation}
here $\hT = T(\hp)$ and $\hV = V(\hx)$. $\htH$ approximates~$\hH$; since $\htH$ must be hermitian,
it has to have an even-power expansion in $\ep$. Here, all operator-products are meant to be from
left to right: $\prod_{i=1}^N \hU_i = \hU_1 \hU_2 \cdots \hU_N$.

Following~\cite{Chin_PRE07}, we will only consider symmetric factorisation schemes
for~(\ref{eq:U_time_ev}) such that the weighting coefficients are either $t_1=0$ and
$v_i=v_{N-i+1}$, $t_{i+1}=t_{N-i+1}$, or $v_N=0$ and $v_i=v_{N-i}$, $t_{i}=t_{N-i+1}$.

Then, according to the Baker-Campbell-Hausdorff formula~\cite{Yoshida_PLA90}, $\htH$ has the form
\begin{eqnarray}
&&\htH = \t1\hT+\v1\hV+\ep^2\ct[\hT\hT\hV]\nn\\
&&\qquad\quad+\ep^2\cv[\hV\hT\hV] +O(\ep^4) \; ,
\label{hopbk}
\end{eqnarray}
where condensed commutator brackets $[\hT\hT\hV]\equiv[\hT,[\hT,\hV]]$,
$[\hT\hV\hT\hV]\equiv[\hT,[\hV,[\hT,\hV]]]$, etc., are used. The coeffi\-ci\-ents $\t1$, $\cv$,
etc., are functions of $\{t_i\}$ and $\{v_j\}$.

By choosing $\{t_i\}$ and $\{v_j\}$ such that $\sum_i t_i = 0 = \sum_j v_j$, we impose that \be
\t1=\v1=0. \label{tvzero} \ee

If we also impose that $\cv=0$, or~$\ct=0$, then the approximate propa\-ga\-tor~(\ref{eq:U_time_ev})
codes for \ns Hamiltonians~$\htH$, either of the form \be \label{eq:HTTV_HVVT} \htH \approx
\hH_{TTV} \propto [\hT\hT\hV] {\rm, \;\; or}\;\; \htH \approx \hH_{VTV} \propto [\hV\hT\hV].  \ee

To summarise, combined separable terms in Eq.~(\ref{eq:U_time_ev}) can emulate specific
nonsepa\-ra\-ble operator products~(\ref{eq:HTTV_HVVT}).

Chin showed in \cite{Chin_PRE07} that the specific symmetric \new{product of nine exponentials}~(\ref{eq:U_time_ev}), with coefficients $ v_0 = -2 ( v_1 + v_2)$, $ t_1 = -t_2$,
$v_2 = -\frac{1}{2}v_1$ and $v_1 = \frac{1}{t_2^2}$, enables us to remove the third order
term $[\hV \hT \hV]$, \new{constituting a `two step forward one step back' scenario. It leads to scaling~$\Delta E \sim {\cal O}(\Delta t^{2/3})$~\cite{Chin_PRE07} (see Sect.~\ref{subsec:WignerKerrNumerics}) and results} in

\begin{eqnarray}\label{eq:Exp9}
\hU_9(\ep) & \approx & e^{\ep v_2 \hV} e^{\ep t_2 \hT} e^{\ep v_1 \hV} e^{\ep t_1 \hT} e^{\ep v_0 \hV} e^{\ep t_1 \hT} e^{\ep v_1 \hV} e^{\ep t_2 \hT} e^{\ep v_2 \hV} \nonumber \\
  & \approx & \exp(\ep^3[\hT \hT \hV] + \ep^5 E_5 + \ep^7 E_7 + \dots) \ . 
\end{eqnarray}
Here, $t_2$ is a free parameter that can be chosen to mini\-mise errors introduced through $E_5$;
following~\cite{Chin_PRE07}, we use $t_2 = -6^{\frac{1}{3}}$.

\subsection{\label{subsec:schrodingerPicture}Propagating a state in the \schr picture}

In split operator techniques, when a propagator~$\exp[\ep(T(\hp)+ V(\hx))]$, with a small time step
$\ep = -i \Delta t/\hbar$, is applied to state~$\psi({\bm x},t)$, we end up applying the sequence of
maps
\begin{eqnarray}\label{eq:unitary_op_f}
\psi(&{\bm x},& t+ \Delta t)  = \exp[\ep(T(\hp)+ V(\hx))]
                             \; \psi({\bm x}, t) \nonumber\\
  & \approx &  \exp(\ep T(\hp)) \exp(\ep V(\hx))
                             \; \psi({\bm x}, t) \nonumber\\
  & = &
              \mathcal{F}_{{\bm p} \rightarrow {\bm x}} \exp(\ep T({\hp})) \;
              \mathcal{F}_{{\bm x} \rightarrow {\bm p}} \exp(\ep V({\hx})) \psi({\bm x}, t) .
\end{eqnarray}
Here, $\cal F$ denotes fast Fourier transforms (and their inverses, in obvious notation) central to
the speedup and numerical stability associated with the use of split opera\-tor techniques.  To give
an example, expression~(\ref{eq:Exp9}) entails the application of~$\cal F$ at eight times per
step~$\ep$.

\section{\label{sec:Kerr}Product terms in Kerr oscillator Hamiltonian}

The single-mode Kerr oscillator, in its simplest form, has the energy of the harmonic oscillator
squared and is therefore ana\-ly\-tically fully solvable. Explicitly, its Hamiltonian has the form
\begin{equation}\label{eq:kerr_h}
  \hH_{\textrm{Kerr}} = \left ( \frac{\hp^2}{2} + \frac{\hx^2}{2} \right )^2
   = \frac{\hp^4}{4} +  \frac{[\hp^2,\hx^2]_+}{4} + \frac{\hx^4}{4} ,
 \end{equation}
 where we used the anti-commutator $[\hat{a},\hat{b}]_+ \equiv \hat{a}\hat{b} + \hat{b}\hat{a}$.
 The quantum Kerr effect comes about due to the self-interaction of photons in nonlinear
 media~\cite{Kirchmair_NAT13}. Its dynamics is non-trivial and periodic with a recurrence time of
 $\tau = \frac{\pi}{\hbar}$; its \ps current follows circles~\cite{Oliva_Kerr_18}.

 We now show that its nonsepa\-ra\-ble terms $[\hp^2, \hx^2]_+$ can be cast into the shape of~$\htH$
 in Eq.~(\ref{hopbk}).  To first order in the time step $\ep = -i \Delta t/\hbar$, the
 Moyal bracket~\cite{Cabrera_PRA15,Oliva_PhysA17} of quantum \ps dynamics agrees with the classical
 Poisson bracket~\cite{Chin_PRE07,Oliva_PhysA17}. We therefore have to hope that the
 commutator~$[\hT \hT \hV]$ in Eq.~(\ref{eq:Exp9}) behaves similarly to the classical Poisson
 bracket--based Lie operators ana\-lysed by Chin~\cite{Chin_PRE07}.

 Following Ref.~\cite{Chin_PRE07}, we therefore try the ansatz of a second order polynomial
 for~$\hT$ and a fourth order polynomial for~$ \hV$. The choices~$\hT=\frac{c_T}{2\sqrt{2}}\hp^2$
 and $ \hV = \frac{c_V}{12}\hx^4$ yield $[\hT \hT \hV] = \frac{c_T^2 c_V}{96} ( \hx^4 
 \hp^4+ \hp^4 
 \hx^4- 2 \; \hp^2
 \hx^4 
 \hp^2)$. With $c_T=1=c_V$ and using Heisenberg's commutation relation
 $[\hp,\hx] =\frac{\hbar}{\rm i}$ this simplifies \new{ (we used
   Mathematica~\cite{Munoz_JPCS16})} to
\begin{eqnarray}\label{eq:TTV}
 [\hT \hT \hV] = -\frac{\hbar^4}{4} - \frac{\hbar^2}{4} [\hp^2,\hx^2]_+, 
\end{eqnarray}
with a real-valued constant term which gives rise to a global phase that can be ignored or
subtracted out.

\new{Alternatively, one can, for example, use
$\hT=\frac{\hp^4}{24} + \frac{\hp^2}{2}  \text{ and } \hV=\frac{\hx^2}{2\sqrt{2}}$
while swapping $\hV \leftrightarrow
\hT$ in expression~(\ref{eq:Exp9}), the formal result is the same~(\ref{eq:TTV}), but numerical
performance can differ slightly~\cite{Bondar_github_nonsep}.}

We have thus established that a propagator using the separable terms $\hT$ and $\hV$ in~$\hU$, in
Eq.~(\ref{eq:U_time_ev}), can gene\-rate a propa\-ga\-tor featuring the product term $[\hp^2,\hx^2]_+$,
which represents the \ns middle term of the Kerr Hamiltonian in Eq.~(\ref{eq:kerr_h}).

\section{\label{sec:WignerKerr}Propagation of mixed states using Wigner's \ps approach}

Instead of \new{limiting ourselves to pure states propagated in the \schr picture}, as in
Sect.~\ref{subsec:schrodingerPicture}, we now study the time evolution of general quantum states
$W(x,p,t)$, in Wigner's \ps representation.

\new{We employ Wigner's representation for the following four reasons: firstly, many dissipative
  systems use coup\-ling terms of product form, so Chin's approach allows us to avoid iterations such
  as those as used in Eqs.~(63) and~(64) of~\cite{Cabrera_PRA15}.  Secondly, the Wigner representation
  describes mixed systems which result from such dissipative couplings. Thirdly, it can be
  efficiently implemented (in Schr\"{o}dinger equation-like form, see below
  and~\cite{Cabrera_PRA15}). Finally, comparison of the quantum with Chin's classical description
  becomes transparent when using the Wigner representation since it describes $W$'s dynamics using
  Moyal brackets~\cite{Moyal_MPCPS49}, the quantum analogue of Poisson brackets:}
\begin{equation}\label{eq:moyal_motion}
    \frac{\partial W}{\partial t} = \{\!\!\{ {H} , W \}\!\!\} = \frac{1}{\rm i \hbar} \hat{\cal G}[W] \; .
\end{equation}
Here, the Hamiltonian $H$, is given by the Wigner transform~\cite{Hancock_EJP04} of $\hH$, which in
the case of the Kerr Hamiltonian~(\ref{eq:kerr_h}) is
$ {H} = \left ( \frac{{p}^2}{2} + \frac{{x}^2}{2} \right )^2 - \frac{\hbar^2}{4}$.  The generator of
motion $\hat{\cal G}$ is the Lie superoperator associated with the Moyal bracket
\cite{Curtright_Treatise2014}, namely

\begin{align}\label{EqMoyalBraket}
    \{\!\!\{ f, g\}\!\!\} &\equiv \frac{f \star g - g \star f}{i\hbar} \notag\\
        &= \frac{2}{\hbar} f(x,p) \sin\left[\frac{\hbar}{2} \left( 
        \overleftarrow{\frac{\partial}{\partial x}} \overrightarrow{\frac{\partial}{\partial p}}
        - \overleftarrow{\frac{\partial}{\partial p}} \overrightarrow{\frac{\partial}{\partial x}}
    \right) \right] g(x,p),
\end{align}
where $\star$ denotes the Groenewold-Moyal product \cite{Groenewold_Phys46,Curtright_Treatise2014}
\begin{align}\label{EqMoyalStar}
    \star &\equiv \exp\left[\frac{i\hbar}{2} \left( 
        \overleftarrow{\frac{\partial}{\partial x}} \overrightarrow{\frac{\partial}{\partial p}}
        - \overleftarrow{\frac{\partial}{\partial p}} \overrightarrow{\frac{\partial}{\partial x}}
    \right) \right] \\
    &= \sum_{n=0}^{\infty} \frac{(i\hbar)^n}{2^n n!} \left( 
        \overleftarrow{\frac{\partial}{\partial x}} \overrightarrow{\frac{\partial}{\partial p}}
        - \overleftarrow{\frac{\partial}{\partial p}} \overrightarrow{\frac{\partial}{\partial x}}
    \right)^n
\end{align}
in which the arrows denote the `direction' of differentiation:
$f\overrightarrow{\frac{\partial}{\partial x}} g = g\overleftarrow{\frac{\partial}{\partial x}} f =
f \frac{\partial}{\partial x} g$.

Taylor's expansion of Moyal's bracket~\eqref{EqMoyalBraket} yields
\begin{align}\label{EqClassicalLimitMoyal}
    \{\!\!\{ f, g\}\!\!\} = \{f, g\} + \mathcal{O}\left( \hbar^n [\mbox{derivatives}^{2n}]\Bigr\vert_{n \geq 2} \right) .
\end{align}
To lowest order, this gives us Poisson's bracket
$\{f, g\} = \frac{\partial f}{\partial x} \frac{\partial g}{\partial p} - \frac{\partial f}{\partial
  p} \frac{\partial g}{\partial x}$ of classical mechanics. We see that in Wigner's representation
the time evolution is formally similar to that in the classical case treated by
Chin~\cite{Chin_PRE07}. \new{We mention in passing that sending $\hbar \downarrow 0$, for instance
  in Eq.~(\ref{eq:generator_op}), (also numerically) implements a classical propagator.}  Wigner's
representation is additionally of interest, because it can be treated efficiently numerically since
Moyal's equation of motion~(\ref{eq:moyal_motion}) can be cast into the form of a Schr\"odinger
equation~\cite{Cabrera_PRA15,Kolaczek_2020}, see next Sect.~\ref{sec:Numerics}.

Equation~\eqref{EqMoyalStar} connects the Wigner transform~\cite{Hancock_EJP04,Curtright_Treatise2014} of
non-commutative Hilbert space operator products $f(\hx, \hp) \cdot g(\hx, \hp)$ with
non-commutative $\star$-products $f(x,p) \star g(x,p)$ on \ps:
\begin{align}
    f(x,p) \star g(x,p) & \Longleftrightarrow f(\hx, \hp) \cdot g(\hx, \hp), \label{EqMoyalStartCorrespondance} \\
    \{\!\!\{ f(x,p), g(x,p) \}\!\!\} & \Longleftrightarrow
    [f(\hx, \hp), g(\hx, \hp)]. \label{EqMoaylCommutaorCorrespondence}
\\\nonumber
\end{align}

\section{\label{sec:Numerics}Numerical considerations}

In Wigner-Weyl transformed variables, we can give $\hat{\cal G}$ of Eq.~(\ref{eq:moyal_motion}) the
explicit form~\cite{Cabrera_PRA15}
\begin{eqnarray}\label{eq:generator_op}
  \hat{\cal G} & = & {H} \left ( \hat{x} - \frac{\hbar}{2}\hat{\theta} , \hat{p}
                     + \frac{\hbar}{2}\hat{\lambda} \right ) - {H} \left ( \hat{x}
                     + \frac{\hbar}{2}\hat{\theta} , \hat{p} - \frac{\hbar}{2}\hat{\lambda} \right ) \quad \\
               & \equiv & \hat{H}_{-,+} - \hat{H}_{+,-}  \quad ,  \label{eq:Hplusminus}
\end{eqnarray}
with the commutation relations \cite{Cabrera_PRA15,Bondar_PRL12}:
\begin{align}
\label{commutation-rels}
 {[} \hat{x} , \hat{p} {]_-} = 0,  \quad
 {[} \hat{x} , \hat{\lambda} {]_-} = i, \quad
 {[} \hat{p} , \hat{\theta} {]_-} = i, \quad
 {[} \hat{\lambda} , \hat{\theta} {]_-} = 0 \; ,
\end{align}
which span a suitable Wigner-Weyl `Hilbert \ps'~\cite{Cabrera_PRA15}. Hence,
\begin{eqnarray}\label{eq:EvolutionCompactlyWritten}
\hat{\cal U} & = & \exp \left( \ep \left ( \hat{H}_{-,+} - \hat{H}_{+,-} \right) \right ) \; ,
\end{eqnarray}
and for time-independent Hamiltonians
\begin{equation}\label{eq:moyal_time_ev}
    W(t) = \exp{\left (-i \frac{t-t_0}{\hbar} \hat{\cal G}\right)}[W(t_0)]  = \hat{\cal U}[W(t_0)] \; .
\end{equation}

We emphasise that in choosing the $(x,\theta)$-represen\-tation~\cite{Cabrera_PRA15},
for Eq.~(\ref{commutation-rels}), using suitable Bopp operators~\cite{bopp1956mecanique}
\begin{align}\label{xThetaRepresentation}
  \hat{x} = x, \quad
  \hat{p} = i \frac{\partial}{\partial \theta}, \quad
  \hat{\lambda} = -i \frac{\partial}{\partial x}, \text{ and }\;
  \hat{\theta} = \theta , \quad 
\end{align}
Eq.~(\ref{eq:moyal_motion}) becomes Schr\"{o}dinger equation-like, making it possible to apply efficient
numerical propagation employing fast Fourier transform
methods~\cite{Cabrera_PRA15,Arnold_SIAM96,Thomann_INMPDE17,Kolaczek_2020}. This is very useful for
systems that cannot be modelled as pure states, such as in the presence of decoherence.

Using $\hat{P}_{\pm} \equiv \hat{p} \pm \frac{\hbar}{2} \hat \lambda$ and
$\hat{X}_{\pm} \equiv \hat{x} \pm \frac{\hbar}{2} \hat \theta$, we can express
$\hat{\cal U}$~(\ref{eq:EvolutionCompactlyWritten}) for the Kerr Hamiltonian~\eqref{eq:kerr_h} as
\begin{widetext}
\begin{subequations}
\begin{eqnarray}
  \hat{\cal U}_{\textrm{Kerr}} & = & \exp{\left[ \frac{\ep}{4} \left ( \hat{P}_+^4 - \hat{P}_-^4
   + \left [ \hat{P}_+^2\hat{X}_-^2 - \hat{P}_-^2\hat{X}_+^2
   + \hat{X}_-^2\hat{P}_+^2 - \hat{X}_+^2\hat{P}_-^2 \right ]
   + \hat{X}_-^4 - \hat{X}_+^4 \right ) \right]}
                                     \label{eq:U_KerrMod} \\
   & = & \exp{\left [ \frac{\ep}{4} \left ( \hat{P}_+^4
   - \hat{P}_-^4 \right ) \right ]}
     \exp {\left[ \frac{\ep}{4} \left( [ \hat{P}_+^2,\hat{X}_-^2]_+
   - [ \hat{P}_-^2,\hat{X}_+^2 ]_+ \right) \right]}
     \exp {\left [ \frac{\ep}{4} \left ( \hat{X}_-^4 - \hat{X}_+^4\right ) \right ]}
    \; +{\cal O}(\ep^2) \;
                                     \label{eq:U_KerrMod2} \; .
\end{eqnarray}
\end{subequations}
According to Eq.~(\ref{eq:TTV}), the appearance of anti-commutators in the middle exponential of
expression~(\ref{eq:U_KerrMod2}) allows us to express the contribution from the central product term
in the Kerr Hamiltonian~(\ref{eq:kerr_h}) as a \emph{single} product of form~(\ref{eq:Exp9}); for an
efficient implementation in Python see~\cite{Bondar_github_nonsep}.
\end{widetext}

\subsection{\label{subsec:WignerKerrNumerics} Error Scaling for Kerr System}

In the following, we set $\hbar = 1$ and use  coherent states 
\begin{equation}\label{eq:W0}
    W \left ( x , p, t=0 \right ) = \frac{1}{\pi} \exp \left ( -(x-x_0)^2 - (p-p_0)^2 \right )
  \end{equation}
as initial states. For an example of their time evolution, see Fig.~\ref{fig:KerrIllustration}.
  
Using exponential propagators (whose action is time-reversible), we confirmed that Chin's approach
preserves the state's norm at machine precision.

We checked for energy and phase stability, varying the time step $\Delta t$. In the case of a classical
system with similar structure Chin reports~\cite{Chin_PRE07}, in accord with Eq.~(\ref{eq:Exp9}),
scaling with order ${\cal O}(\ep^5)/{\cal O}(\ep^3) \sim {\cal O}(\Delta t^{2/3})$; this is roughly what
we observed here, in the quantum Kerr case, as well.

The period of our Kerr system is $\pi$, see Fig.~\ref{fig:KerrIllustration}~(c). As a proxy for
phase drift, associated with this algorithm, we determine the wave function overlap at recurrence
times~$\tau=\pi$ and find that
$1- \langle W_{\rm exact}(0) | W(\tau) \rangle \sim {\cal O}(\Delta t^{1.2})$, a scaling better than
that of the energy fluctuations (we distinguish between $W_{\rm exact}$ and the numerically
propagated distribution $W$). For this we could not find a quantitative explanation. We also
numerically propa\-ga\-ted~$\psi$ and confirmed that
$1 - |\langle \psi_{\rm exact}(0)| \psi(\pi)\rangle|^2 \sim {\cal O}(\Delta t^{1.2})$ follows the
same scaling as the Wigner function propagator. To make sure there is nothing special about one
complete revolution we also checked for~$\tau=\pi/2$,
$1- \langle W_{\rm exact}(\pi/2) | W(\tau) \rangle \sim {\cal O}(\Delta t^{1.2})$, the behaviour is
the same. Alternatively, $\langle W_{\rm exact}(0)-W(\tau)|W_{\rm exact}(0)- W(\tau) \rangle^{1/2}$
and max${}_{x,p}|W_{\rm exact}(x,p,0)-W(x,p,\tau)| $, used as overlap measures, yield error scaling
$\sim {\cal O}(\Delta t^{2/3})$ and, again, behave the same whether running for one complete
revolution or otherwise. For details see Ref.~\cite{Bondar_github_nonsep}.

\subsection{\label{subsec:7Exp} Modification of Chin's expression}

In the Kerr-case studied here, it is possible to modify Chin's expression~(\ref{eq:Exp9}) by
removing first and last terms yiel\-ding the approximation

\begin{figure*}[t]
   \begin{minipage}[h]{1.999\columnwidth}
     \includegraphics[width=0.235\textwidth]{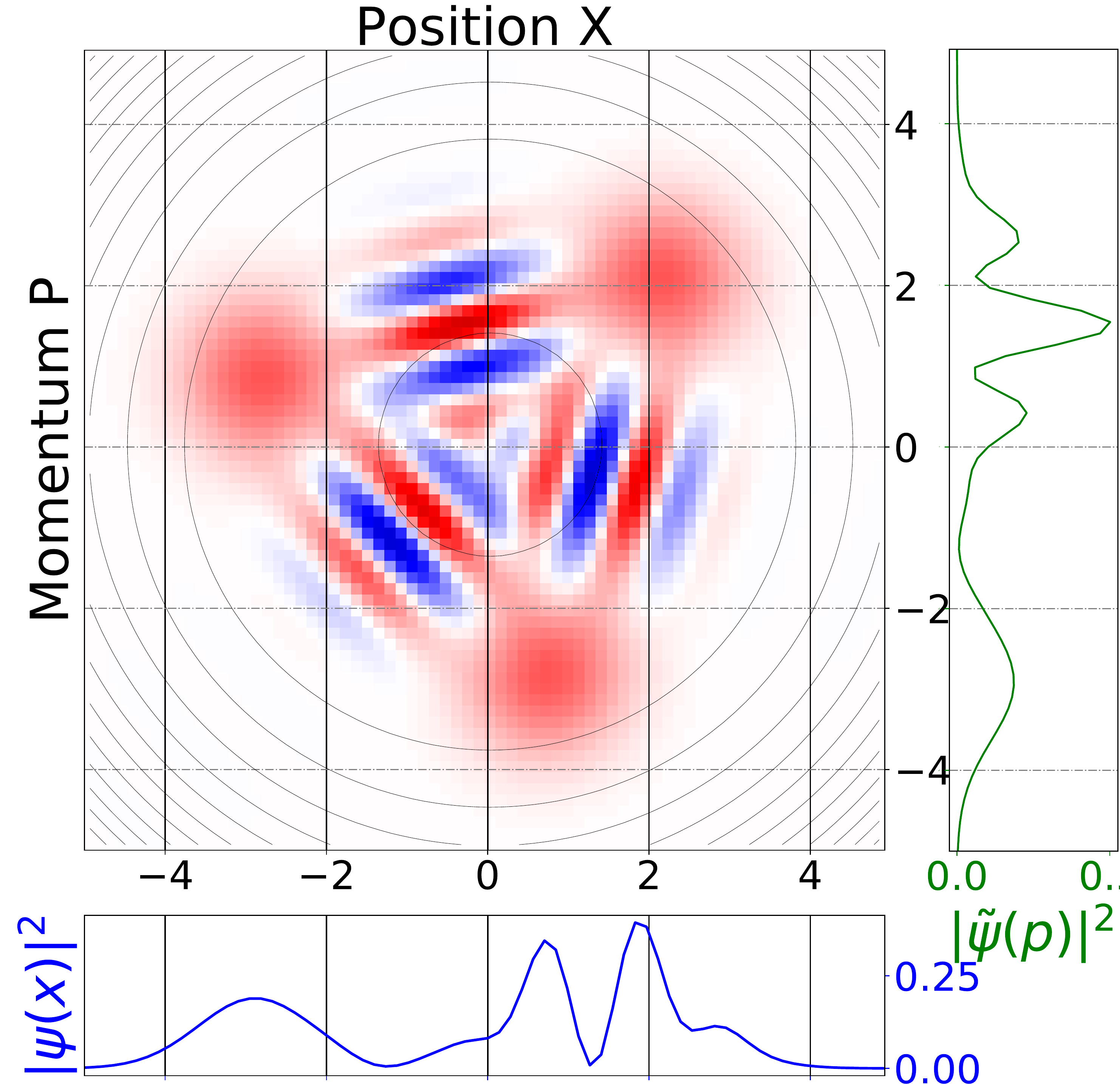}
     \put(-106,120){\rotatebox{0}{{\bf (a)} $\substack{t=1.047 \approx \tau/3 \\ (\Delta t=0.0001)}$ }}     \includegraphics[width=0.235\textwidth]{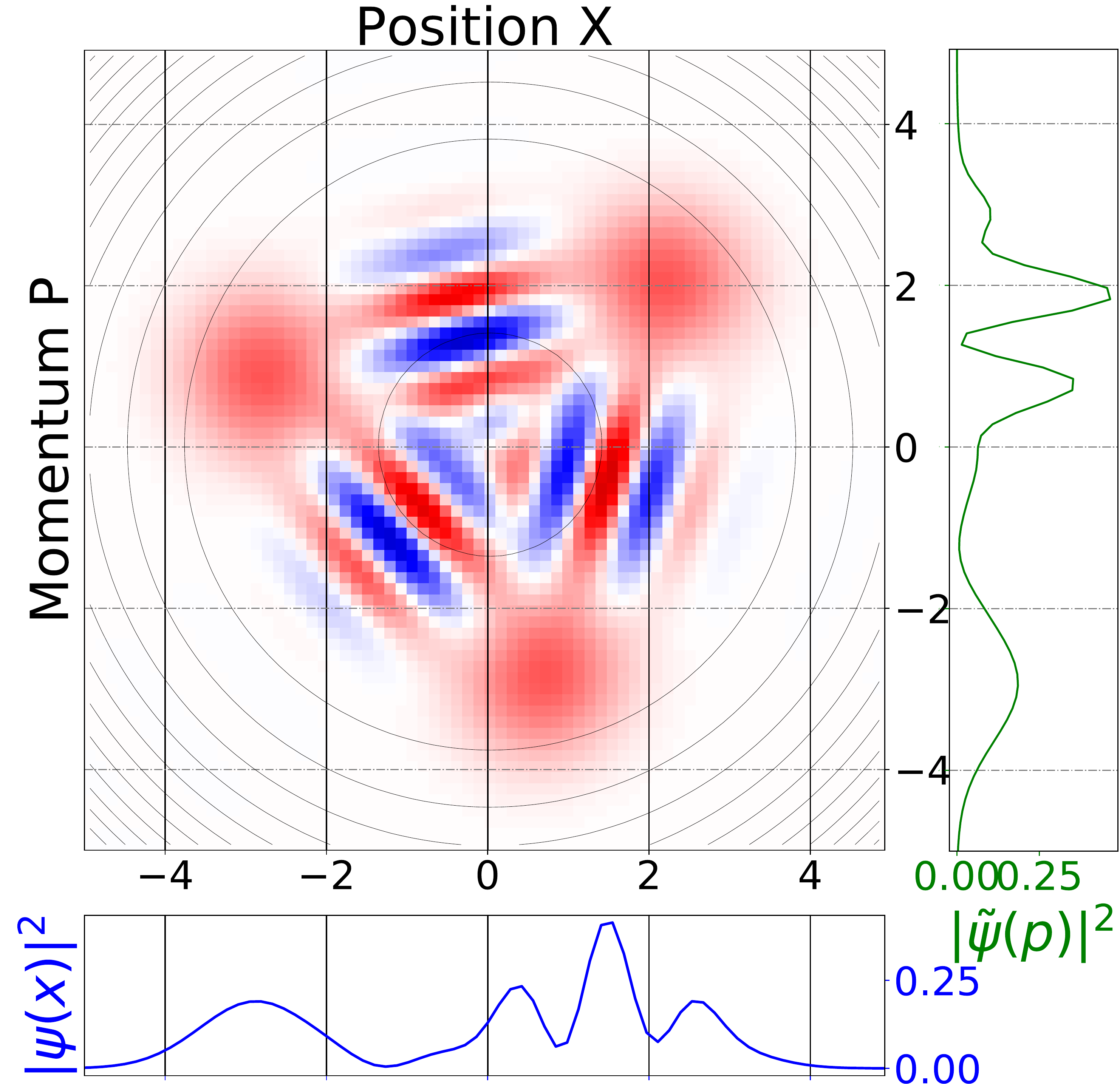}
     \put(-106,120){\rotatebox{0}{{\bf (b)} $\substack{t=2.095 \approx \frac{2}{3} \tau \\ (\Delta t=0.0001)}$ }} \includegraphics[width=0.235\textwidth]{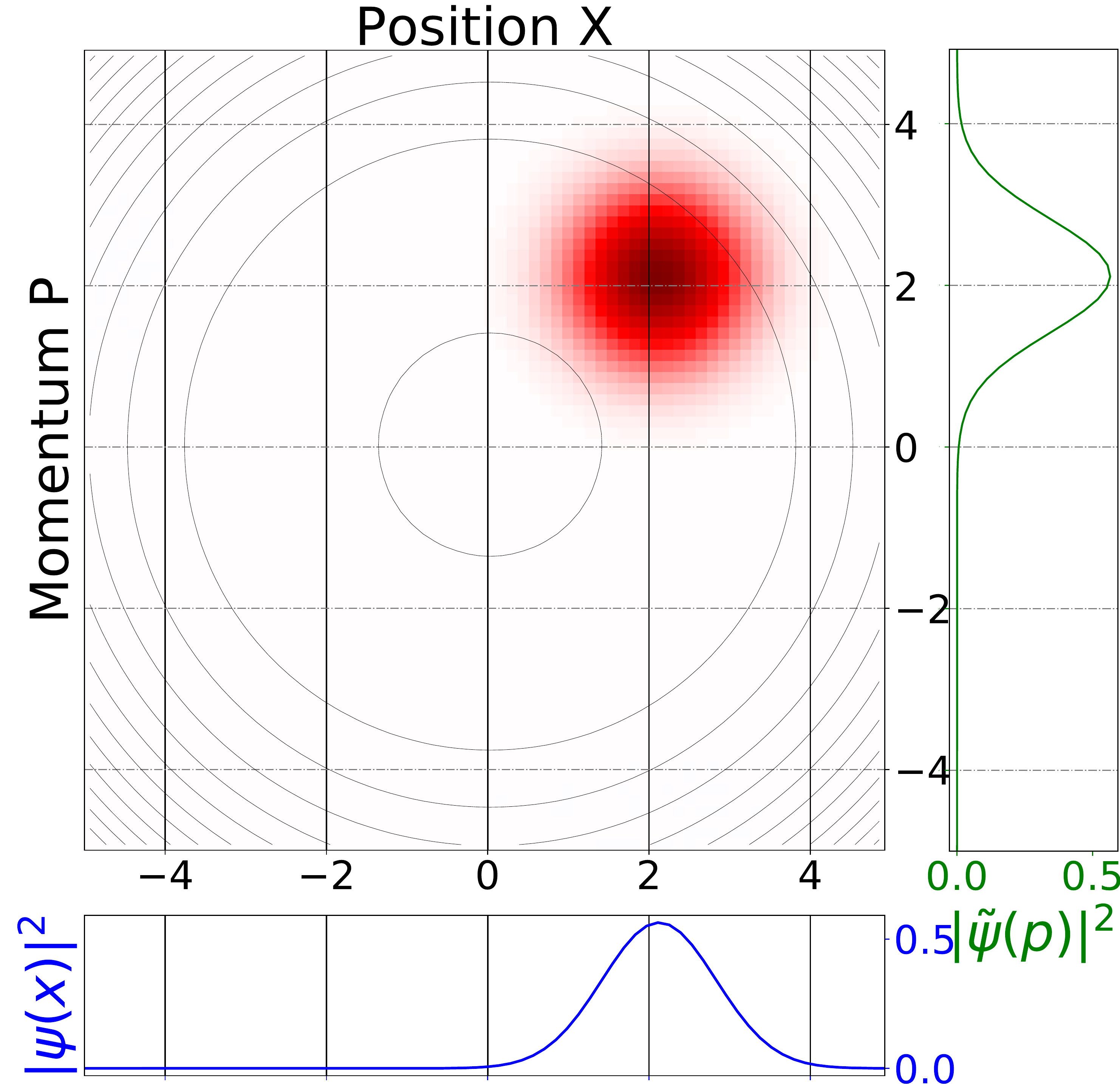}
     \put(-106,120){\rotatebox{0}{{\bf (c)} $\substack{t=3.142 \approx \tau \\ (\Delta t=0.0001)}$ }}
    \includegraphics[width=0.0475\textwidth]{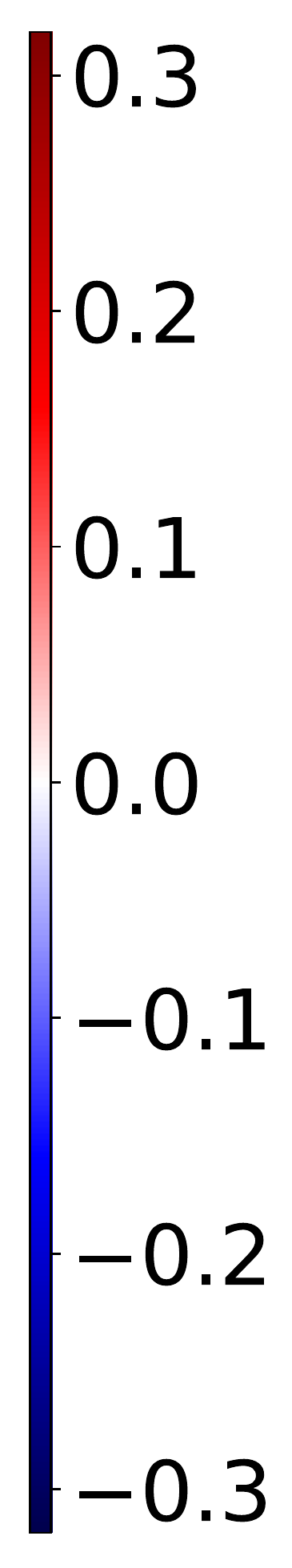} \includegraphics[width=0.235\textwidth]{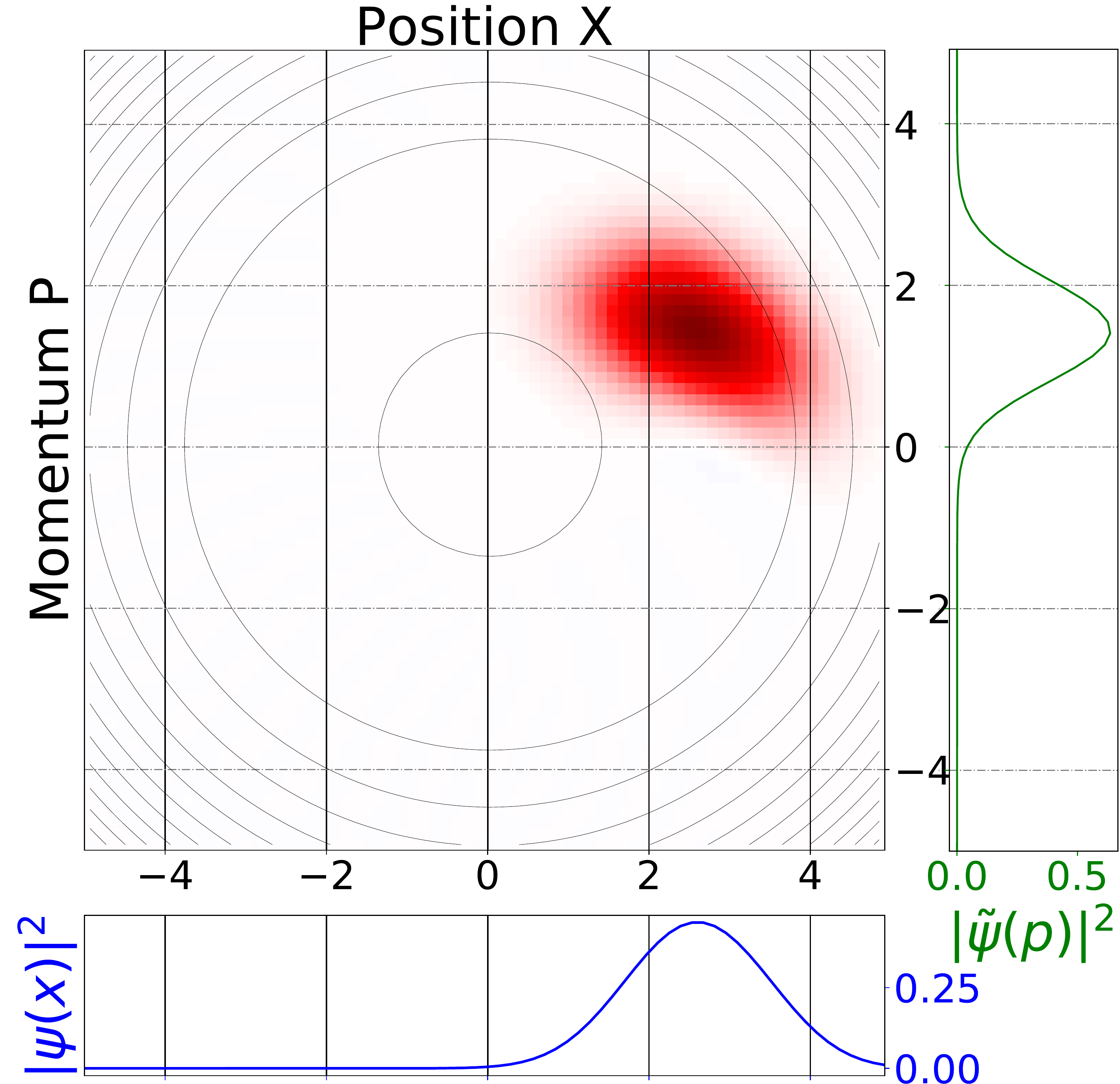}
     \put(-106,120){\rotatebox{0}{{\bf (d)} $\substack{t=3.142 \approx \tau \\ (\Delta t=0.0015)}$ }}
     \\ \includegraphics[width=0.23\textwidth]{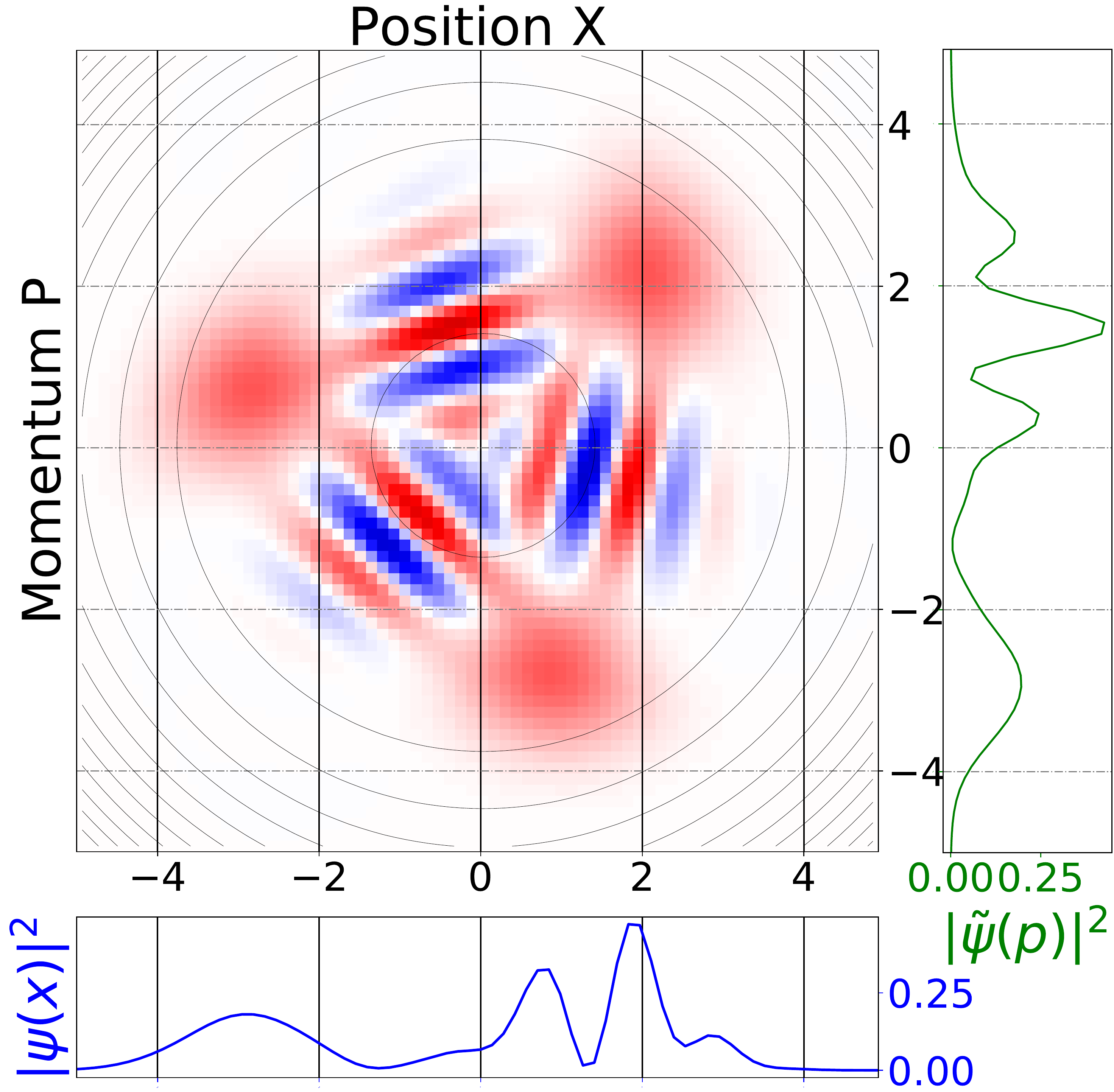} \includegraphics[width=0.23\textwidth]{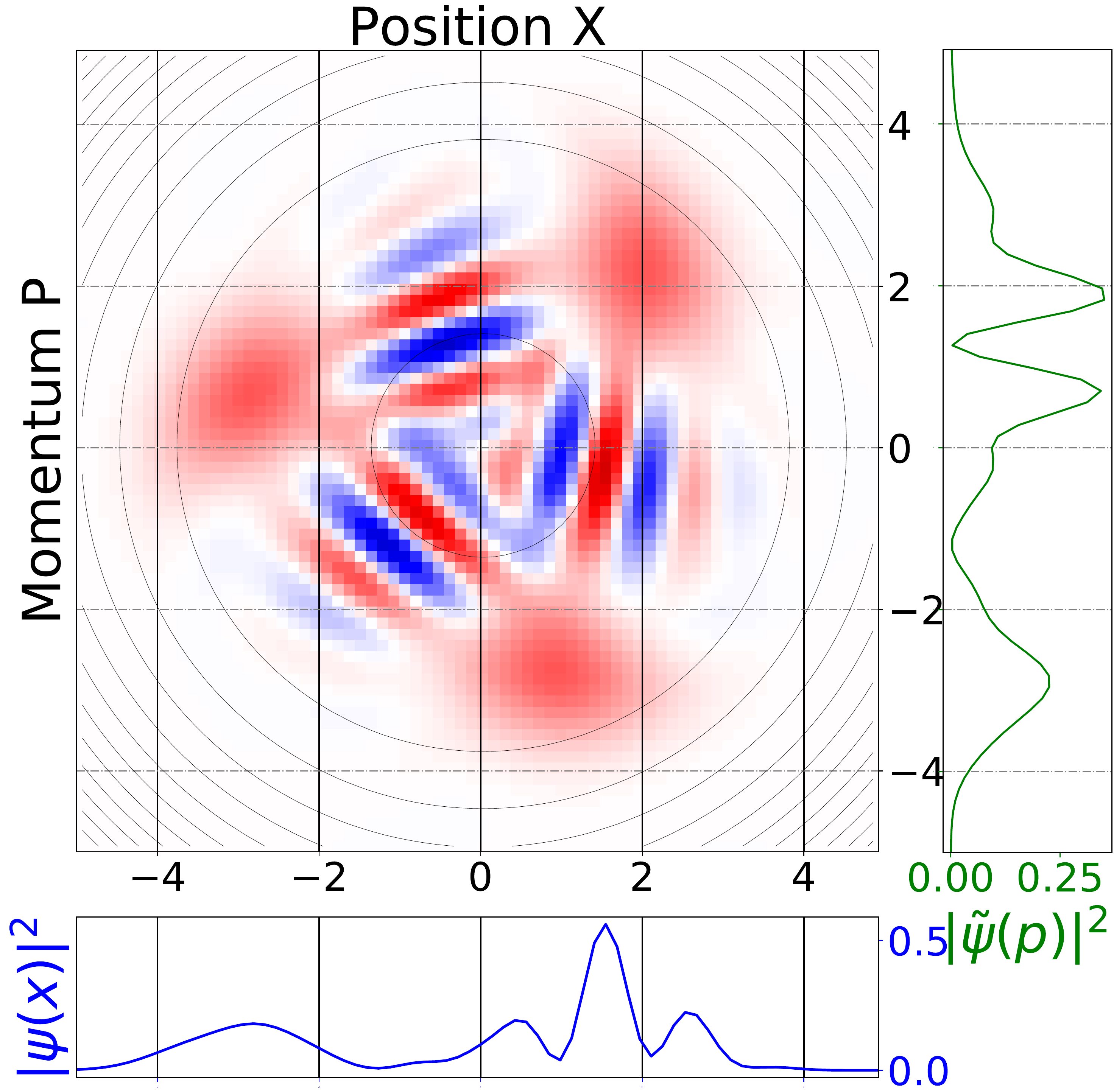} \includegraphics[width=0.23\textwidth]{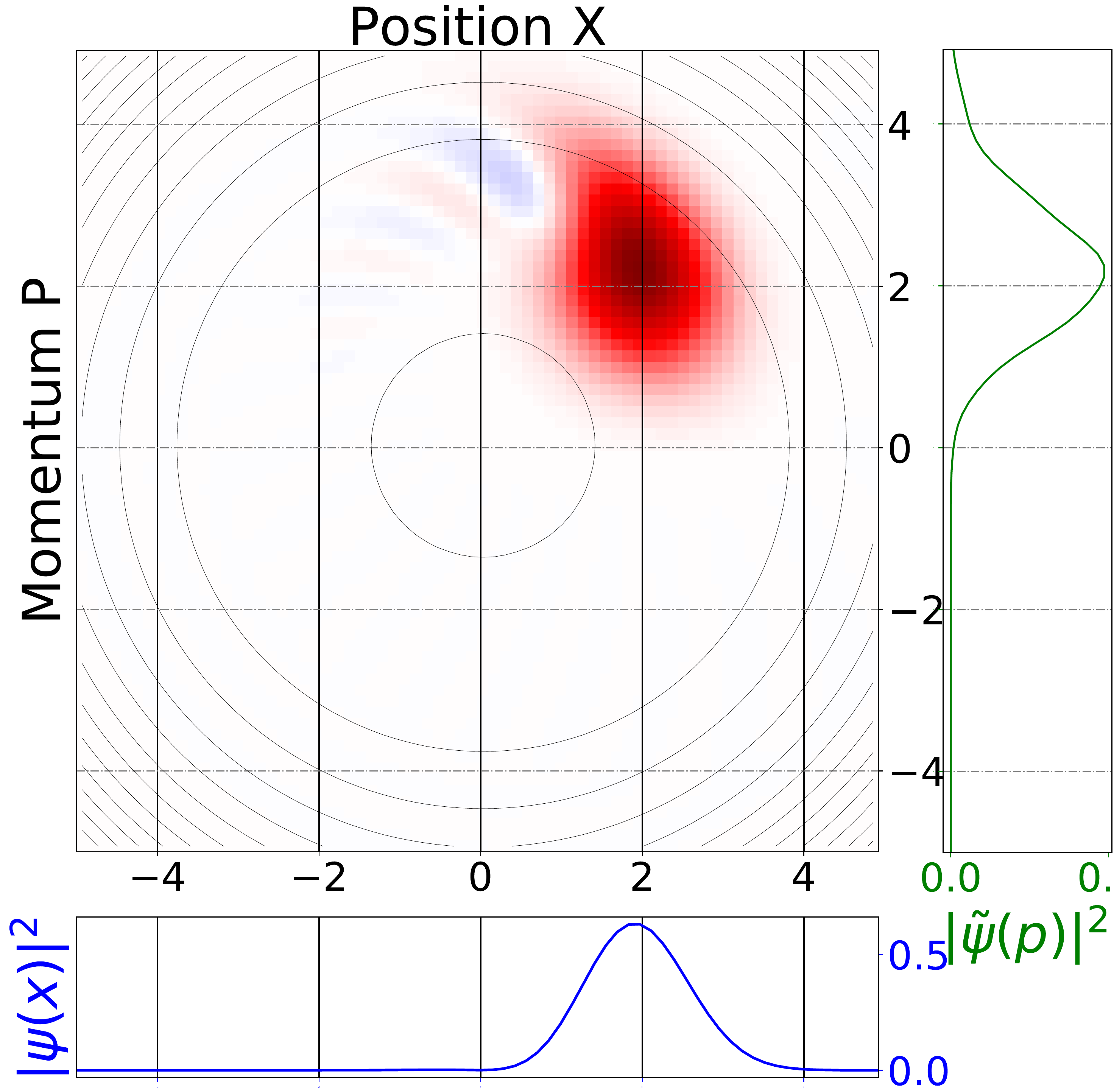}
         \includegraphics[width=0.0475\textwidth]{ColorBar.pdf} \includegraphics[width=0.235\textwidth]{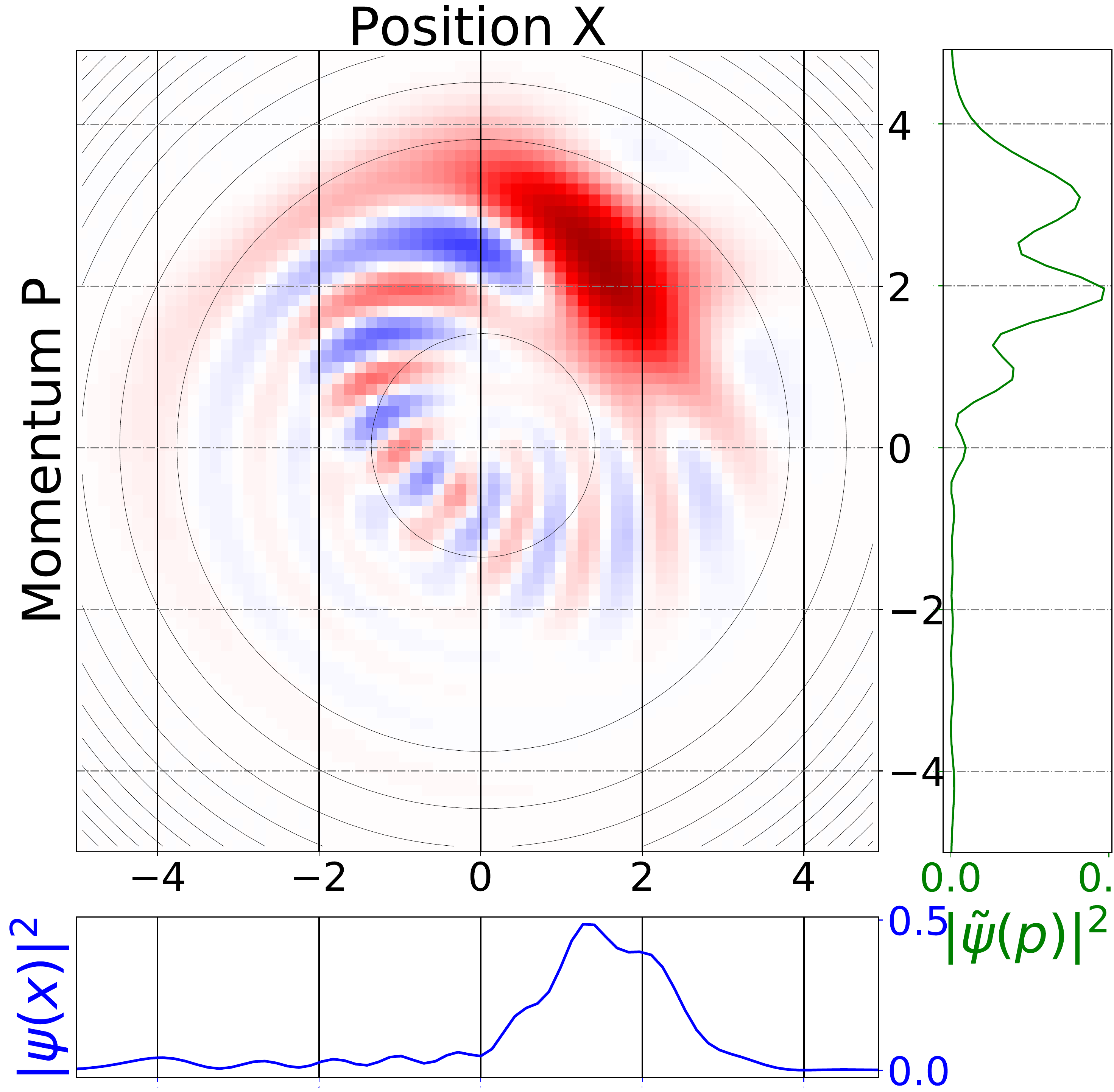}
     \caption{Kerr system~(\ref{eq:kerr_h}) evolves an initial coherent state~(\ref{eq:W0}) with
           amplitude `3' [$W_0(x_0=3/\sqrt{2}, p_0=3/\sqrt{2})$] using timesteps $\Delta t=0.0001$ in
           columns {\bf (a)} to {\bf (c)} and $\Delta t=0.0015$ in column {\bf (d)}. We observe that the
           initial state recurs at the system's recurrence time~$\tau = \pi $~columns {\bf (c)} and
           {\bf (d)} and fractional revivals of the initial states with approximate three-fold
           symmetry~\cite{Oliva_Shear_19,Averbukh_PLA89,Robinett_PR04} form at times $\tau / 3$ and
           $2 \tau / 4$, see columns {\bf (a)} and {\bf (b)}, respectively.  The same parameters are
           used in top and bottom row, except that the top row employs the more accurate propagator
           $\hU_9$ of Eq.~(\ref{eq:Exp9}) whereas the bottom row uses $\hU_7$ of
           Eq.~(\ref{eq:Exp7}), the numerical errors become more pronounced with both, increasing
           propagation time and increasing timesteps.}
    \label{fig:KerrIllustration}
   \end{minipage}
  \end{figure*}

  \begin{subequations}
  \begin{eqnarray}\label{eq:Exp7}
  \hU_7(\ep) & = & e^{\ep t_2 \hT} e^{\ep v_1 \hV} e^{\ep t_1 \hT}
  e^{\ep v_0 \hV} e^{\ep t_1 \hT} e^{\ep v_1 \hV} e^{\ep t_2 \hT} \qquad \\
  & = & \hU_9(\ep) - \frac{\ep^3}{t_2^3} [\hV\hV\hT] + {\cal O}(\ep^5) .
  \end{eqnarray}
\end{subequations}

We use the modified coefficients $ v_0 = -2 v_1$, $ t_1 = -t_2$, $v_1 = \frac{1}{t_2^2}$, and
again~$\hT=\frac{c_T}{2\sqrt{2}}\hp^2$ and $ \hV = \frac{c_V}{12}\hx^4$, with the final
result~\cite{Munoz_JPCS16}

\begin{equation}\label{eq:TTV7}
\begin{aligned}
  \hU_7(\ep) &= 1 + \ep^3\left( -\frac{\hbar^4}{4} - \frac{\hbar^2}{4} [\hp^2,\hx^2]_+
  + \frac{\hbar^2}{9 \sqrt{2} \, t_2^3} \hx^6 \right) + {\cal O}(\ep^5). \qquad
\end{aligned}
\end{equation}

This is similar to result~(\ref{eq:TTV}). We have to compensate for the unwanted term in $\hx^6$, by
subtracting $ \frac{\hbar^2}{9 \sqrt{2} \, t_2^3} \hx^6$ from potential terms in
Eqs.~(\ref{eq:unitary_op_f}) and~(\ref{eq:U_KerrMod2}), but gain the advantage of having to
numerically calculate fewer terms.

Whether a form like~(\ref{eq:Exp7}) that is as useful as~(\ref{eq:TTV7}) can be found in the general 
case, we do not know at this stage. We emphasise that $\hU_7$ has fewer product terms and runs a little 
faster but also performs worse than~$\hU_9$ of Eq.~(\ref{eq:Exp9}) in absolute terms, 
see Fig.~\ref{fig:KerrIllustration}. The errors in energy and phase both scale roughly with
${\cal O}(\Delta t^{2/3})$, similarly to and worse than in the case of~$\hU_9$, respectively.

For more details consult our code~\cite{Bondar_github_nonsep}, further discussions of such questions is
beyond the scope of this work.

\section{\label{sec:ChinGeneral}Chin's approach is general}

In order to show that Chin's approach is generally applicable, let us prove
\begin{theorem}\label{TheoremGenralDecomposition}
  Any polynomial $\new{\Pi}(\hx, \hp)$ of $\hx$ and $\hp$ can be written as a finite linear combination of
  $[\hx^n\hx^n\hp^m] \equiv [\hx^n, [\hx^n, \hp^m]]$ and
  $[\hp^m\hx^n\hp^m] \equiv [\hp^m, [\hx^n, \hp^m]]$.
\end{theorem}
\begin{proof}
  Let us provide a constructive proof. A polynomial \new{$\Pi(\hx, \hp)$ is
    Weyl-transformed~\cite{Hancock_EJP04}} to $P(x,p)$ in \ps, according to
  Eq.~\eqref{EqMoyalStartCorrespondance}. Assume $x^N p^M$ is its leading term, namely, $P(x,p)$ is
  a polynomial of order $(N,M)$.
  
  Via Eq.~\eqref{EqMoaylCommutaorCorrespondence}, a double commutator $[\hx^n \hx^n \hp^m]$
  corresponds to $\{\!\!\{ x^n, \{\!\!\{x^n, p^m \}\!\!\} \}\!\!\}$. In fact, the latter are
  polynomials because the Moyal bracket~\eqref{EqMoyalBraket} is obtained by differentiating its
  arguments. According to Eq.~\eqref{EqClassicalLimitMoyal}, the leading term of the polynomial
  $\{\!\!\{ x^n, p^m \}\!\!\}$ is $x^{n-1}p^{m-1}$; hence, the leading term of
  $\{\!\!\{ x^n, \{\!\!\{x^n, p^m \}\!\!\} \}\!\!\}$ is $x^{2n-2} p^{m-2}$. Likewise, a double
  commutator $[\hp^m \hx^n \hp^m]$ corresponds to the polynomial
  $\{\!\!\{ p^m, \{\!\!\{x^n, p^m \}\!\!\} \}\!\!\}$ with the leading term of order
  $x^{n-2} p^{2m-2}$.

The set of polynomials 
\begin{align}
  \Big[ \{\!\!\{ x^n, \{\!\!\{x^n, p^m \}\!\!\} \}\!\!\}, \{\!\!\{ p^m, \{\!\!\{x^n, p^m \}\!\!\} \}\!\!\}
  \Big]_{n=1,2,\ldots, N + 2}^{m=1,2,\ldots, M+2}
\end{align}
is linearly independent and large enough to span the set of polynomials of order $(N,M)$, including $P(x,p)$.
\end{proof}

We observe that in the above proof all Moyal brackets $\{\!\!\{.,. \}\!\!\} $ can be substituted by
Poisson brackets $\{.,.\} $ whilst leaving the argument intact: Chin's approach applies to polynomial
classical hamiltonians as well.

Theorem~\ref{TheoremGenralDecomposition} prescribes how any polynomial quantum Hamiltonian can be
decomposed into the Hamiltonians of the form \eqref{eq:HTTV_HVVT}. Hence, Chin's algorithm is very
general.

\section{\label{sec:conclusion}Conclusion}

We have shown that Chin's method~\cite{Chin_PRE07} for the propa\-gation of classical \ns
Hamiltonians can be adopted to quantum systems.  Chin's method is general, and therefore allows for
the universal treatment of \ns Hamiltonians using split-operator techniques. \new{Chin's method should be especially well suited for numerical simulations of large open quantum systems using stochastic Schr\"{o}dinger equations~\cite{jacobs_straightforward_2006} since their errors scale poorly.}

\begin{acknowledgments}
\new{We thank both reviewers for their many thoughtful suggestions.}
  D.I.B. was supported by by the W. M. Keck Foundation and Army Research Office (ARO) (grant
  W911NF-19-1-0377; program manager Dr.~James Joseph). The views and conclusions contained in this
  document are those of the authors and should not be interpreted as representing the official
  policies, either expressed or implied, of ARO or the U.S. Government. The U.S. Government is
  authorized to reproduce and distribute reprints for Government purposes thank both reviewers for their many thoughtful suggestions.notwithstanding any
  copyright notation herein.
\end{acknowledgments}

\section*{Data Availability Statement}

The codes developed for the current study are available at~\cite{Bondar_github_nonsep}.

\bibliography{Ole_Bibliography}




\end{document}